\documentclass[11pt,a4paper]{article}
\pdfoutput=1 
\usepackage[utf8]{inputenc}
\usepackage[margin=1in]{geometry}
\usepackage{amsfonts}
\usepackage{amsmath}
\usepackage{amssymb}
\usepackage{amsthm}
\usepackage{float}
\usepackage[font=small,labelfont=bf]{caption} 
\allowdisplaybreaks
\usepackage{bbold}

\usepackage{physics}
\usepackage{xspace}

\usepackage{bold-extra} 

\usepackage{hyperref}
\usepackage[svgnames]{xcolor}
\hypersetup{colorlinks={true},urlcolor={MidnightBlue},linkcolor={MidnightBlue},citecolor=[named]{MidnightBlue}}

\usepackage{microtype}
\usepackage[capitalise,nameinlink,noabbrev]{cleveref}

\usepackage{tikz}
\usetikzlibrary{arrows.meta}

\usepackage{doi}

\usepackage{enumerate}

\theoremstyle{definition}
\newtheorem{definition}{Definition}

\theoremstyle{plain}
\newtheorem{theorem}{Theorem}[section]
\newtheorem{lemma}[theorem]{Lemma}

\theoremstyle{remark}

\Crefname{claim}{Claim}{Claims}

\def\np/{\textup{\textsf{NP}}}
\def\tfnp/{\textup{\textsf{TFNP}}}
\def\ppad/{\textup{\textsf{PPAD}}}
\def\fixp/{\textup{\textsf{FIXP}}}

\def\pcircuit{\textup{\textsc{Pure-Circuit}}\xspace}

\newcommand{\garbo}{\ensuremath{\bot}\xspace}

\newcommand{\val}[1]{\boldsymbol{\mathrm{x}}[#1]}
\newcommand{\valonly}{\boldsymbol{\mathrm{x}}}
\newcommand{\zerone}[1]{\{0,1\}^{#1}}
\newcommand{\Du}[1]{\Delta u_{#1}}

\newcommand{\PURE}{\textup{\textsf{PURIFY}}\xspace}
\newcommand{\NOR}{\textup{\textsf{NOR}}\xspace}

\newcommand{\supp}{\textup{\textsf{supp}}\xspace}
\newcommand{\nash}[1]{\ensuremath{#1}-NE\xspace}
\newcommand{\wn}[1]{\ensuremath{#1}-WSNE\xspace}

\newcommand{\tildu}{\tilde{u}}

\newcommand{\eps}{\ensuremath{\varepsilon}\xspace}

\newcommand{\vba}{\ensuremath{\vb{a}}\xspace}
\newcommand{\vbai}{\ensuremath{\vb{a}_{-i}}\xspace}
\newcommand{\vbs}{\ensuremath{\vb{s}}\xspace}
\newcommand{\vbsi}{\ensuremath{\vb{s}_{-i}}\xspace}

\newcommand{\bfx}{\ensuremath{\mathbf{x}}\xspace}
\newcommand{\bfs}{\ensuremath{\mathbf{s}}\xspace}

\newcommand{\OtherPurifyNode}[1]{c_{#1}}
\newcommand{\CentralPurifyNode}{{d}}

\title{Tight Inapproximability of Nash Equilibria\\ in Public Goods Games}

\author{
\begin{tabular}{cc}
& \\
 \textbf{Jérémi Do Dinh} & \textbf{Alexandros Hollender} \\
   \small{École Polytechnique Fédérale de Lausanne}
  & \small{University of Oxford}\\
 \href{mailto:jeremi.dodinh@epfl.ch}{\small{\texttt{jeremi.dodinh@epfl.ch}}} & \href{mailto:alexandros.hollender@cs.ox.ac.uk}{\small{\texttt{alexandros.hollender@cs.ox.ac.uk}}} \\
&
\end{tabular}
}

\date{}

\begin{document}

\maketitle

\begin{abstract}
\noindent We study public goods games, a type of game where every player has to decide whether or not to produce a good which is \emph{public}, i.e., neighboring players can also benefit from it. Specifically, we consider a setting where the good is indivisible and where the neighborhood structure is represented by a directed graph, with the players being the nodes. Papadimitriou and Peng (2023) recently showed that in this setting computing mixed Nash equilibria is \ppad/-hard, and that this remains the case even for $\eps$-well-supported approximate equilibria for some sufficiently small constant $\eps$. In this work, we strengthen this inapproximability result by showing that the problem remains \ppad/-hard for any non-trivial approximation parameter $\eps$.
\end{abstract}


\section{Introduction}
A good is said to be \emph{public} if everyone can simultaneously benefit from it, and not only the person or organization who expended money (or effort) for its provision. Examples of public goods include clean air and water, information on the internet (e.g., Wikipedia, open-access scientific papers), or television broadcasts.

Public goods give rise to interesting game theoretic settings due to the \emph{free-riding} phenomenon: an agent can benefit from a good without having to pay for it. Indeed, consider a setting in which every agent can decide whether to pay for the provision of the public good, but if at least one agent pays for the good, then everyone has access to it. Clearly, every agent is happier if the public good is available as opposed to unavailable, but ideally they would prefer some other agent to pay for its provision.

In many applications, the public good does not benefit everyone but rather only the agents in the proximity of the good. This is for example the case for clean air or water. A very natural way to model this ``locality'' is by assuming that the agents (or players) are nodes in a graph, and the neighborhood structure encodes which agents will benefit from a public good produced by a given agent~\cite{BramoulleK07-public-goods}. In other words, if a certain agent decides to produce the good, then all its neighbors, as well as the agent itself, will benefit from it. Various versions of such public goods games on graphs have been extensively studied~\cite{BramoulleK07-public-goods,GaleottiGJVY09-public-goods,DallAstaPR11-best-shot,LopezPintado13-public-goods,BramoulleKD14-strategic-networks,Allouch15-public-goods,BramoulleK16-games-networks,YuZBV20-public-goods,GilboaN22-public-goods,BayerKS23-public-goods,KlimmS23-public-goods,Gilboa23-public-goods}.

In this work, we focus on a particularly simple setting with indivisible public goods. We consider binary action public goods games with the $\max$ utility function, which are sometimes also called \emph{best-shot games}. In such a game, every player $i$ chooses an action $a_i \in \{0,1\}$, where $a_i = 1$ corresponds to producing the public good, and $a_i = 0$ to not producing it. The utility of a player is the maximum over the strategies of all the players in its neighborhood (including itself) minus the cost (or price) of producing the good. Formally
$$u_i(a_1,\dots,a_n) = \max_{j \in N(i) \cup \{i\}} \{a_j\} - p \cdot a_i,$$
where $N(i)$ denotes the neighborhood of player $i$, and $p \in (0,1)$ is the price of producing the good. Unless stated otherwise, we fix $p=1/2$ in the rest of this discussion.

When the graph is undirected, pure Nash equilibria always exist, and it is not hard to see that they exactly correspond to \emph{maximal} independent sets of the graph~\cite{BramoulleK07-public-goods}. In particular, a pure Nash equilibrium can be computed efficiently.

In some settings, interactions are not necessarily symmetric. In that case, it is more suitable to model the interactions using a \emph{directed} graph. A good example is the spread of information in social networks, which often have a directed structure. When the graph is directed, a pure Nash equilibrium is no longer guaranteed to exist~\cite{BramoulleK16-games-networks} and Papadimitriou and Peng~\cite{PapadimitriouP23-public-goods} have shown that it is \np/-complete to decide whether one exists (and thus also to compute one).

In view of the non-existence of pure equilibria, it is natural to consider \emph{mixed} Nash equilibria, which are guaranteed to exist even in the directed setting. Unfortunately, it turns out that mixed equilibria are also hard to compute, and even to approximate. Papadimitriou and Peng~\cite{PapadimitriouP23-public-goods} prove that it is \ppad/-hard to compute an $\eps$-well-supported approximate Nash equilibrium, for some sufficiently small (unspecified) constant $\eps > 0$. To show this, they introduce a new type of game called \emph{threshold game} and prove that computing an approximate equilibrium in such a game reduces to the problem of computing an approximate equilibrium of a public goods game. The hardness of computing approximate equilibria for threshold games is then shown by reducing from the \emph{generalized-circuit} problem~\cite{ChenDT09-Nash} which is known to be \ppad/-hard to solve approximately for some sufficiently small (unspecified) constant $\eps > 0$~\cite{Rubinstein18-Nash-inapproximability}.

Recent work by Deligkas et al.~\cite{DeligkasFHM22-pure-circuit} has introduced a new tool, called the \pcircuit problem, for proving \ppad/-hardness for \emph{explicit} large constant approximation factors. In particular, using \pcircuit, they were able to show that computing an $\eps$-approximate equilibrium of a threshold game is \ppad/-hard for any $\eps < 1/6$. This is in fact tight, as there exists a simple polynomial-time algorithm that computes a $1/6$-approximate equilibrium.

Plugging this strengthened inapproximability result for threshold games into the reduction of Papadimitriou and Peng yields an improved inapproximability result for public goods games. Namely, it follows that computing an $\eps$-well-supported approximate Nash equilibrium of a binary public goods game is \ppad/-hard for any $\eps < 1/48$. In particular, the reduction from threshold games to public goods games loses a factor $8$ in the approximation term.

\paragraph*{Our contribution.}
In this paper, we prove that computing an $\eps$-well-supported approximate Nash equilibrium of a binary public goods game remains \ppad/-hard for any $\eps < 1/2$. More generally, for any price parameter $p \in (0,1)$, we show \ppad/-hardness for any $\eps < \min \{p,1-p\}$ (\cref{thm:pgg-elementary}). This result is tight, as a trivial $\min \{p,1-p\}$-well-supported approximate equilibrium can easily be found (see \cref{lemma:trivial-bounds}).

To prove our result, it is necessary to bypass threshold games, since the existing inapproximability for that problem is already tight. Instead, we provide a direct reduction from the \pcircuit problem to our problem of interest.

\section{Preliminaries}
In this section, we introduce the notions needed to state and prove our main theorem. Our results are related to two domains of study: algorithmic game theory and the study of \ppad/, an important subclass of \tfnp/. Therefore this section is organized into two subsections. Firstly, we define public goods games, along with the related game theoretic notions and the associated computational problems. Secondly, we define the \pcircuit problem, which serves as the \ppad/-complete problem used to demonstrate the main hardness result.

\subsection{Public Goods Games}
We now define the important notions concerning public goods games, along with the associated game theoretic notions and the relevant computational problems. 

\begin{definition}[Game]\label{def:game}
    A \textit{game} is defined by a set of players $ [n] $ and a set of actions (pure strategies) $ A_i $ available to each player $i \in [n]$. Each player additionally has a utility function $ u_i: A_1 \times \dots \times A_n \rightarrow \mathbb{R} $, by which they obtain a payoff depending on the action choice of all the players. A \textit{pure strategy profile} is a tuple $\vba \in A_1 \times \dots \times A_n$, representing the choice of actions for each of the players.
\end{definition}

Throughout this work we use the standard game theoretic notation, with:
\[
\vbai = (a_1, \dots, a_{i-1}, a_{i+1}, \dots, a_n),
\]
and
\[
(k, \vbai) = (a_1, \dots, a_{i-1},k, a_{i+1}, \dots, a_n).
\]

\begin{definition}\label{def:indiv-pgg}
	An \textit{indivisible public goods game in a directed graph} is given by a vertex set $ V = [n] $ corresponding to the players, a set $ E $ of directed edges, as well as a price $ p \in (0,1) $.
    Each player has the same action (pure strategy) set $ A = \{0,1\} $, corresponding to the choice of whether or not to produce the good. The \textit{payoff} of player $i$ when playing strategy $k$ is:
\[
u_i(k, \vbai) = \begin{cases}
    1-p & \text{if } k = 1\\
    1 & \text{if $k = 0 $ and $ \exists j \in N(i)$ s.t. $a_j = 1$} \\
    0 & \text{otherwise,}
\end{cases}
\]
where the \textit{neighborhood} for any player $ i \in V $ is defined to be:
\[ N(i) = \{j \mid (j, i) \in E \}.
\]
\end{definition}

\begin{definition}[Mixed strategy]\label{def:mixed-strategy-pgg}
    A \textit{mixed strategy} is a player's choice of a distribution over their set of actions. In a public goods game, this corresponds to a value $s_i \in [0,1]$ assigned to each player representing the probability that player $i$ will produce the good. 

    A \textit{strategy profile} is a choice of a mixed strategy (just \textit{strategy} from hereon) for each of the players. In a public goods game, it is a tuple $\vbs \in  [0,1]^n  = \Sigma $ ($\Sigma$ denoting the set of all possible strategy profiles).
\end{definition}

\begin{lemma}
    Given a mixed strategy profile $ \vbs \in [0,1]^n $, the expected payoff of player $i$ is $\tildu_i(\vbs) = s_i \cdot \tildu_i(1, \vbsi) + (1-s_i) \cdot \tildu_i(0, \vbsi)$, where for $k \in \{0,1\}$:
	 \[
	 \tildu_i(k, \vbsi)	=  
	 \begin{cases}
	 	1-p & k = 1 \\
	 	1 - \prod_{j \in N(i)}(1 - s_j) & k = 0.
	 \end{cases}
	 \]
\end{lemma}

\begin{proof}
We consider the two cases corresponding to the choice of pure strategy $k$ of player $i$. 
\paragraph{Case: $k = 1$.} In this scenario, player $i$ is unaffected by the choice of strategy of the players in $N(i)$, and instead, the player benefits from the good while also incurring the cost of production. Therefore, the expected payoff is $1-p$.
\paragraph{Case: $k = 0$.} In this scenario, player $i$'s expected payoff is the probability that at least one of the players in $N(i)$ produces the good. This corresponds to the complement of the probability that none of the players in $N(i)$ produce the good, which is $\prod_{j \in N(i)}(1 - s_j)$.
\end{proof}

\begin{definition}[$\eps$-best-response]
    Given $\eps \geq 0$ and a strategy profile for the other players $ \vbsi $, we say that a strategy $s \in [0,1]$ is an $\eps$-\textit{best-response} of player $ i $ to $ \vbsi $, if it maximizes, up to $\eps$, the expected payoff given that the other players play $\vbsi$:
\[
\tildu(s,\vbsi) \geq \max_{s' \in [0,1]}\tilde{u}(s', \vbsi) - \eps. 
\]
If $\eps = 0$, then we simply use the term \emph{best-response}.
Note that there always exists a pure strategy $k \in \{0,1\}$ that is a best-response.
\end{definition}

\begin{definition}[Nash equilibria and approximations]\label{def:neq}
    Consider a strategy profile $ \mathbf{s} = (s_1, \dots, s_n) \in [0,1]^n $. We call $\mathbf{s}$ a \textit{Nash equilibrium} if no player can improve their expected payoff by choosing a different strategy. In other words, every player is playing a best-response:
\[
\forall i \in [n], \forall s'_i \in [0,1], \tilde{u}_i(\vbs) \geq \tilde{u}_i(s'_i, \vbsi).
\]
We say that $\mathbf{s}$ is an \eps-\textit{Nash equilibrium} (\nash{\eps}) if each of the players' expected payoffs is within \eps of their best-response payoff, i.e., they all play an $\eps$-best-response. That is:
\[
\forall i \in [n], \tilde{u}_i(\mathbf{s}) \geq \max_{k \in \{0,1\}} \tilde{u}(k, \vbsi) - \eps.
\]
Finally, $\mathbf{s}$ is an $ \eps $-\textit{Well Supported Nash Equilibrium} (\wn{\eps}) if each player places a positive probability exclusively on pure strategies that are $ \eps $-best-responses. That is:
\[
\forall i \in [n], \forall k \in \supp{(s_i)}, \quad u_i(k, \vbsi) \geq \max_{k' \in \{0,1\}} \tilde{u}(k', \vbsi) - \eps,
\] 
where $ \supp{(s_i)} = \Big \{j \in \zerone{} : (s_i)_j >0 \Big \} $. 
\end{definition}

\begin{definition}[Utility difference]
The utility difference for a player $ i $ is the difference in the expected payoff for a player between their choice of pure strategies:
\[
\Delta u_i = \tildu_i(1, \vbsi) - \tildu_i(0, \vbsi) = \prod_{j \in N_i}(1 - s_j) - p.
\]
\end{definition}

\paragraph{Observation.}In an \eps-WSNE, $ \Delta u_i $ dictates precisely which strategy player $ i $ would play. In particular:
\begin{itemize}
	\item if $ \Delta u_i > \eps $, then strategy $ 1 $ is the only pure \eps-best-response for player $ i $, and therefore $ s_i = 1 $.
	\item if $ \Delta u_i < -\eps $, then strategy $ 0 $ is the only pure \eps-best-response for player $ i $, and therefore $ s_i = 0 $.
	\item if $ \Delta u_i \in [-\eps, \eps] $, then both strategies $ 0 $ and $ 1 $ are \eps-best-responses for player $ i $, and therefore $ s_i $ can be anything in $ [0,1] $.
\end{itemize}

\begin{lemma}\label{lemma:trivial-bounds}
    If $ \eps \geq \min\{p, 1-p\} $, the game admits a trivial \eps-WSNE.
\end{lemma}
\begin{proof}
    To demonstrate this, we consider the two sub-cases of the scenario where  $ \eps \geq \min\{p, 1-p\} $ and show that a trivial equilibrium point exists in these situations. Namely:
\begin{itemize}
	\item if $ \eps \geq p = \min\{p, 1-p\} $ then it is impossible for $ \Delta u_i < -\eps $ and therefore an \eps-WSNE is trivially obtained by having all the players produce the good.
	\item if $ \eps \geq 1 - p = \min\{p, 1-p\} $ then it is impossible for $ \Delta u_i > \eps $ and therefore an \eps-WSNE is trivially obtained with none of the players producing the good. 
\end{itemize}
Thus, we obtain a trivial equilibrium in both cases.
\end{proof}

\subsection{\pcircuit}
For the purposes of the main result presented here, we consider instances of \pcircuit that make use of the \NOR and \PURE gates. We now provide the corresponding definition of \pcircuit. 
\begin{definition}[\pcircuit \cite{DeligkasFHM22-pure-circuit}]\label{def:pcircuit}
An instance of the \pcircuit problem is given by a set of nodes $V=[n]$ and a set $G$ of gates. Each gate $g \in G$ is of the form $g = (T,u,v,w)$ where $u,v,w \in V$ are distinct nodes and $T \in \{\NOR, \PURE\}$ is the type of the gate, with the following interpretation.
\begin{itemize}
	\item If $T=\NOR$, then $u$ and $v$ are the inputs of the gate, and $w$ is its output.
	\item If $T=\PURE$, then $u$ is the input of the gate, and $v$ and $w$ are its outputs.
\end{itemize}
We require that each node is the output of exactly one gate.

A solution to an instance $(V,G)$ is an assignment $\valonly: V \to \{0,1,\garbo\}$ that satisfies all the gates, i.e., for each gate $g=(T,u,v,w) \in G$ we have:
\begin{itemize}
	\item if $T=\NOR$, then $\valonly$ satisfies
			\begin{align*}
				\val{u} = \val{v} = 0 \implies \val{w} = 1\\
				(\val{u} = 1) \lor (\val{v} = 1) \implies \val{w} = 0
			\end{align*}
	
	\item if $T=\PURE$, then $\valonly$ satisfies
			\begin{align*}
				\{\val{v}, \val{w}\} \cap \{0,1\} \neq \emptyset\\
				\val{u} \in \{0,1\} \implies \val{v} = \val{w} = \val{u}
			\end{align*}
\end{itemize}
\end{definition}
We now also state the central theorem from \cite{DeligkasFHM22-pure-circuit} related to the \pcircuit problem, and necessary to our main result.

\begin{theorem}[{{\cite{DeligkasFHM22-pure-circuit}}}]
  The \pcircuit problem with \NOR and \PURE gates is \ppad/-complete.  
\end{theorem}

\section{Reduction from \pcircuit}
In \cite[Section 4]{DeligkasFHM22-pure-circuit} reductions from the \pcircuit problem were demonstrated for the \eps-generalized-circuit problem and for finding various equilibria points in polymatrix games, as well as threshold games. The last reduction is most relevant to this work, firstly because the setting for threshold games is similar to the setting of public goods games (we can think of public goods games as a multiplicative version of threshold games), but also because it was precisely threshold games which were used as an intermediate problem to show \ppad/-hardness of public goods games in~\cite{PapadimitriouP23-public-goods}.

We proceed to provide the main result, to establish the optimal bound for $ \eps $, which constitutes the main contribution of this work.
\begin{theorem}\label{thm:pgg-elementary}
	For any $ p \in (0,1) $ and any constant $ \eps < \min\{p, 1-p\} $, it is \ppad/-hard to find an \wn{\eps} of binary public good games on directed graphs with price $p$. 
\end{theorem}
\begin{proof}
	We show that the problem is \ppad/-hard for any constant $ \eps < \min\{p, 1-p\} $ through a reduction from \pcircuit with the use of the \NOR and \PURE gates.

    \textbf{Construction.}
	Consider an instance $(V,G)$ of \pcircuit. For each node $u \in V$ we introduce a corresponding node $u \in V'$, where $V'$ is the set of nodes in the public goods graph. Each of the gates in the \pcircuit instance is implemented by adding corresponding edges and additional auxiliary nodes as shown in \cref{fig:gadgets-pgg}. The \NOR gadget is very simple, and it only introduces two edges between the corresponding nodes. The \PURE gadget is a generalization of a similar construction appearing in \cite{DeligkasFHM22-pure-circuit} in the context of threshold games. It introduces additional auxiliary nodes $\CentralPurifyNode, \OtherPurifyNode{1}, \dots, \OtherPurifyNode{2l+1}$, as well as the edges denoted in the figure. We set
 $$ l = \left\lceil \frac{\log(p-\eps)}{\log(p + \eps)} \right\rceil. $$

	Let $ \bfs $ be an \wn{\eps} of the public goods game. For any $u \in V'$, $  s_u \in [0,1] $ denotes the probability of player/node $u$ producing the good. We obtain an assignment $ \bfx $ to the \pcircuit instance as follows:
	\[
	\forall u \in V, \ \bfx[u] = \begin{cases}
		s_u & \text{if}\ s_u \in \zerone{} \\
		\bot & \text{otherwise.}
	\end{cases}
	\]

    \textbf{Correctness.}
	We will now demonstrate that the assignment $ \bfx $ obtained in this manner satisfies the \pcircuit gate constraints.
	
	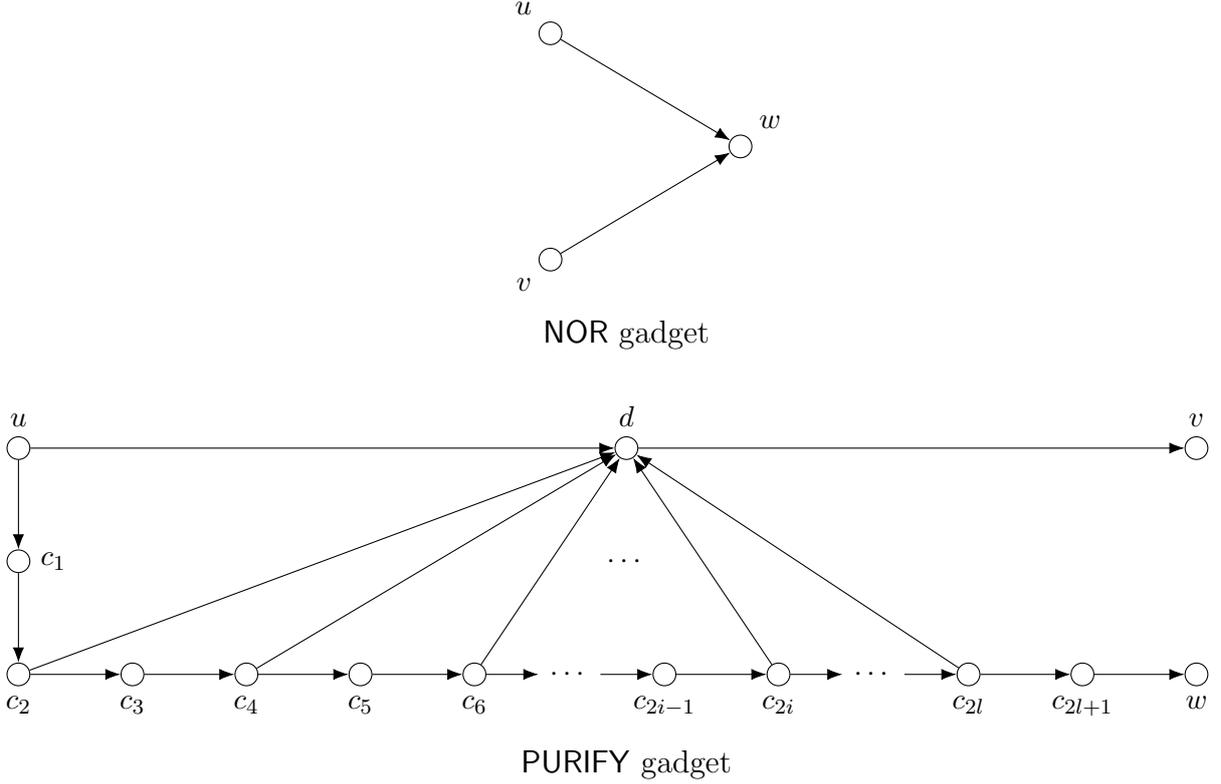
\begin{figure}
		\centering
  \begin{tikzpicture}[roundnode/.style={circle, draw=black, inner sep=0, minimum size=3mm}]

 	\node[roundnode,label=above left:{$u$}] (NORu) at (7,7.5) {};
	\node[roundnode,label=below left:{$v$}] (NORv) at (7,4.5) {};
	\node[roundnode,label=above right:{$w$}] (NORw) at (9.5,6) {};
	
	\draw[-{Latex[length=2mm]}] (NORu) -- (NORw);
	\draw[-{Latex[length=2mm]}] (NORv) -- (NORw);

 	\node at (8,3.5) {\large \NOR gadget};
	
	\node[roundnode,label=above:{$u$}] (PUREu) at (0,2) {};
	
	\node[roundnode,label=right:{$\OtherPurifyNode{1}$}] (PUREa) at (0,0.5) {};
	\node[roundnode,label=below:{$\OtherPurifyNode{2}$}] (PUREb) at (0,-1) {};
	\node[roundnode,label=below:{$\OtherPurifyNode{3}$}] (PUREc) at (1.5,-1) {};
	\node[roundnode,label=below:{$w$}] (PUREw) at (15.5,-1) {};
	\node[roundnode,label=below:{$\OtherPurifyNode{4}$}] (PUREc1) at (3,-1) {};
	\node[roundnode,label=below:{$\OtherPurifyNode{5}$}] (PUREc2) at (4.5,-1) {};
	\node[roundnode,label=below:{$\OtherPurifyNode{6}$}] (PUREc3) at (6,-1) {};
	\node[roundnode, draw=none] (PUREc4) at (7,-1) {};
	\node[roundnode, draw=none] (PUREc5) at (7.25,-1) {$\cdots$};
	\node[roundnode, draw=none] (PUREc6) at (7.5,-1) {};
	\node[roundnode,label=below:{$\OtherPurifyNode{2i-1}$}] (PUREc7) at (8.5,-1) {};
	\node[roundnode,label=below:{$\OtherPurifyNode{2i}$}] (PUREc8) at (10,-1) {};
	\node[roundnode,label=below:{$\OtherPurifyNode{2l+1}$}] (PUREw1) at (14,-1) {};
	\node[roundnode,label=below:{$\OtherPurifyNode{2l}$}] (PUREw2) at (12.5,-1) {};
	\node[roundnode, draw=none] (q_dots) at (8, 0.5) {$\cdots$}; 
	
	\node[roundnode, draw=none] (PUREw5) at (11,-1) {};
	\node[roundnode, draw=none] (PUREw4) at (11.25,-1) {$\cdots$};
	\node[roundnode, draw=none] (PUREw3) at (11.5,-1) {};
	
	\node[roundnode,label=above:{$\CentralPurifyNode$}] (PUREd) at (8,2) {};
	\node[roundnode,label=above:{$v$}] (PUREv) at (15.5,2) {};
	
	\draw[-{Latex[length=2mm]}] (PUREu) -- (PUREa);
	\draw[-{Latex[length=2mm]}] (PUREa) -- (PUREb);
	\draw[-{Latex[length=2mm]}] (PUREb) -- (PUREc);
	\draw[-{Latex[length=2mm]}] (PUREc) -- (PUREc1);
	\draw[-{Latex[length=2mm]}] (PUREc1) -- (PUREc2);
	\draw[-{Latex[length=2mm]}] (PUREc2) -- (PUREc3);
	\draw[-{Latex[length=2mm]}] (PUREc3) -- (PUREd);
	\draw[-{Latex[length=2mm]}] (PUREc3) -- (PUREc4);
	\draw[-{Latex[length=2mm]}] (PUREc6) -- (PUREc7);
	\draw[-{Latex[length=2mm]}] (PUREc7) -- (PUREc8);
	\draw[-{Latex[length=2mm]}] (PUREc8) -- (PUREd);

	\draw[-{Latex[length=2mm]}] (PUREw2) -- (PUREw1);
	\draw[-{Latex[length=2mm]}] (PUREw2) -- (PUREd);
	\draw[-{Latex[length=2mm]}] (PUREw1) -- (PUREw);
	
	\draw[-{Latex[length=2mm]}] (PUREc8) -- (PUREw5);
	\draw[-{Latex[length=2mm]}] (PUREw3) -- (PUREw2);
	
	\draw[-{Latex[length=2mm]}] (PUREu) -- (PUREd);
	\draw[-{Latex[length=2mm]}] (PUREb) -- (PUREd);
	\draw[-{Latex[length=2mm]}] (PUREc1) -- (PUREd);
	\draw[-{Latex[length=2mm]}] (PUREd) -- (PUREv);

    \node at (8,-2.2) {\large \PURE gadget};
\end{tikzpicture}
		\caption{The gadgets used for the public goods instance in \cref{thm:pgg-elementary}.}
		\label{fig:gadgets-pgg}
	\end{figure}
	
	\paragraph{\NOR gates.}
	A $ (\NOR, u,v,w) $ gate is simulated by the \NOR gadget depicted in \cref{fig:gadgets-pgg}.
	\begin{itemize}
		\item If $ s_u = s_v = 0 $, then in particular it means that $ \Delta u_w = (1-s_u)(1-s_v) - p = 1 - p > \eps $. This implies that action $ 1 $ is the only pure \eps-best-response for $ w $, and therefore in an \wn{\eps} we will have that $ s_w = 1 $.
		\item If $ s_u = 1 $ or $ s_v = 1 $, then we have that $ \Delta u_w = (1-s_u)(1-s_v) - p = - p < -\eps $. This implies that action $ 0 $ is the only pure \eps-best-response for $ w $, and therefore in an \wn{\eps} we will have that $ s_w = 0 $.
		\item In any other remaining case, whatever strategy $ w $ plays will map back to a value satisfying the \NOR gate.
	\end{itemize}
	
		\paragraph{\PURE gates.}
    A $ (\PURE, u,v,w) $ gate is simulated by the \PURE gadget depicted in \cref{fig:gadgets-pgg}.
	\begin{itemize}
		\item  If $ s_u \in \{0,1\} $, then for all $ j \in [l] $ it holds that $ s_{\OtherPurifyNode{2j}} = s_u $. This is easy to show by repeatedly applying a similar argument to the one used above for the \NOR gadget. Similarly, we also deduce that $ s_w = s_u $. Furthermore, it holds that:
		\[ \Du{\CentralPurifyNode} = (1-s_u)\prod_{j \in [l]}(1-s_{\OtherPurifyNode{2j}}) - p  =(1-s_u)^{l+1} - p = \begin{cases}
			1-p \ > \eps & \text{if } s_u = 0 \\
			\ -p \quad < -\eps & \text{if } s_u = 1 
		\end{cases}. \]
	It follows from this that $s_\CentralPurifyNode \in \{0,1\}$ and $ s_\CentralPurifyNode = 1 - s_u$, and thus $ s_v = s_u $. Together with $ s_w = s_u $, this means that the gate constraint is satisfied for this case.  
	\item If $ s_u \notin \zerone{} $, then we show that if $ s_w \notin \zerone{} $ it must hold that $ s_v = 1 $. Note that if $ s_w \notin \zerone{} $, then we must have that $ \Du{w} \in [-\eps, \eps] $. This in turn implies that $ s_{\OtherPurifyNode{2l+1}} \in [1-p - \eps, 1-p + \eps] \subset (0,1) $, since $ \eps < \min\{p, 1-p\} $. By repeatedly applying this argument we deduce that $ s_u, s_{\OtherPurifyNode{2j}} \in  [1-p - \eps, 1-p + \eps]  $, for all $ j \in [l] $. As a result, we have that:
	\[
	\Du{\CentralPurifyNode} = \underbrace{(1-s_u)\prod_{j \in [l]}(1-s_{\OtherPurifyNode{2j}})}_{\in \left [(p- \eps)^{l+1}, (p + \eps)^{l+1}\right ]} - p < (p + \eps) ^ l - p \leq (p + \eps) ^ {\frac{\log(p-\eps)}{\log(p + \eps)} } - p = -\eps. 
	\]
    where we used the fact that $0 < p - \eps \leq p + \eps < 1$.
	It follows from above that $ 0 $ is the only pure $\eps$-best-response for $ \CentralPurifyNode $, and therefore $ s_\CentralPurifyNode = 0 $, which in turn implies that $ s_v = 1 $, as desired.
	\end{itemize}

As shown above, we have that an \wn{\eps} in the game correctly encodes the gate constraints imposed by the \NOR and \PURE gates.
\end{proof}

\section{Open Problems}

We conclude with some natural open questions that arise from our work:
\begin{itemize}
    \item Is it possible to obtain similar tight inapproximability results for $\eps$-Nash equilibria?
    \item The problem of computing an \emph{exact} Nash equilibrium lies in the class \fixp/ introduced by Etessami and Yannakakis~\cite{EtessamiY10-FIXP}. Is it \fixp/-complete? Or does the problem always admit rational equilibria?
    \item 
    Is it possible to prove the hardness result with the use of a gadget whose size does not depend on how close $\eps$ is to $\min\{p,1-p\}$? In particular, is it possible to prove hardness for the problem with $p=1/2$ and for any $\eps < 1/2$ for graphs with \emph{bounded} degree? Are positive results possible?

\end{itemize}

\bigskip

\subsubsection*{Acknowledgements}
We thank Mika Göös for helpful discussions in the early stages of this project.
The second author was supported by the Swiss State Secretariat for Education, Research and Innovation (SERI) under contract number MB22.00026.

\bibliographystyle{alphaurl}
\bibliography{references}

\end{document}